\documentclass[a4paper,USenglish,numberwithinsect]{lipics}

\usepackage{amsfonts}

\newcommand{\Nat}{I\!\!N}
\newcommand{\Rat}{I\!\!R}

\usepackage{color}

\usepackage[vlined,ruled]{algorithm2e}

\newcommand{\Rank}{\mbox{$\mathit{rank}$}}
\newcommand{\Select}{\mbox{$\mathit{select}$}}

\title{Space-Efficient Plane-Sweep Algorithms}

\author[1]{Amr Elmasry}
\author[2]{Frank Kammer}
\affil[1]{Department of Computer Engineering and Systems\\
Alexandria University, Alexandria 21544, Egypt\\
  \texttt{elmasry@alexu.edu.eg}}
\affil[2]{Institut f\"ur Informatik, Universit\"at Augsburg\\
  86135 Augsburg, Germany\\
  \texttt{kammer@informatik.uni-augsburg.de}}
\authorrunning{A.\ Elmasry and F.\ Kammer} 
\Copyright{Amr Elmasry and Frank Kammer}

\subjclass{E.1 Data Structures F.2.2 Nonnumerical Algorithms and Problems, I.1.2 Analysis of Algorithms, I.3.5 Geometric Algorithms}

\keywords{closest pair, line-segments intersection, Klee's measure}

\serieslogo{}%
\volumeinfo%
  {Billy Editor and Bill Editors}%
  {2}%
  {Conference title on which this volume is based on}%
  {1}%
  {1}%
  {1}%
\EventShortName{}
\DOI{10.4230/LIPIcs.xxx.yyy.p}%

\begin{document}

  \maketitle{}

\begin{abstract}
We introduce space-efficient plane-sweep algorithms for basic planar geometric problems.
It is assumed that the input is in a read-only array of $n$ items and that the available workspace is $\Theta(s)$ bits, where $\lg n \leq s \leq n \cdot \lg n$.
Three techniques that can be used as general tools in different space-efficient algorithms are introduced and employed within our algorithms.
In particular, we give an almost-optimal algorithm for finding the closest pair among a set of $n$ points that runs in $O(n^2/s + n \cdot \lg s)$ time. 
We also give a simple algorithm to enumerate the intersections of $n$ line segments that runs in $O((n^2/s^{2/3}) \cdot \lg s + k)$ time, where $k$ is the number of intersections. 
The counting version can be solved in $O((n^2/s^{2/3}) \cdot \lg s)$~time.
When the segments are axis-parallel, we give an $O((n^2/s) \cdot \lg^{4/3} s + n^{4/3} \cdot \lg^{1/3} n)$-time algorithm for counting the intersections, and an algorithm for enumerating the intersections that runs in $O((n^2/s) \cdot \lg s \cdot \lg \lg s + n \cdot \lg s + k)$ time, where $k$ is the number of intersections. We finally present an algorithm that runs in $O((n^2/s + n \cdot \lg s) \cdot \sqrt{(n/s) \cdot \lg n})$ time to calculate Klee's measure of axis-parallel rectangles.
\end{abstract}

\section{Introduction}
Because of the rapid growth of the input data sizes in current applications, algorithms that are designed 
to efficiently utilize space are becoming even more important than before.
One other reason for the need for space-efficient algorithms is the
limitation in the memory sizes that can be deployed to modern embedded systems. 
Therefore, many algorithms have been developed with the objective of optimizing the time-space product. 

Several models of computation have been considered for the case when writing in the input area is restricted. 
The objective of a space-efficient algorithm is to optimize the amount of extra space needed to perform its task.
In the \emph{multi-pass streaming model} \cite{MunP80} the input is assumed to be held in a read-only sequentially-accessible media, and the goal would be to optimize the number of passes an algorithm makes over the input.
In the \emph{read-only random-access model} \cite{Fre87}---the model that we
consider in this paper---the input is assumed to be stored on a read-only
randomly-accessible media, and arithmetic operations on operands that fit in
a computer word are assumed to take constant time each.
Here, optimizing the number of arithmetic operations would be the target.
Another model, introduced in \cite{BroIKMMT04}, allows the input to be permuted but not destroyed.
For a variant of the latter model, called the \emph{restore model} \cite{ChaMR14}, 
the input array is allowed to be modified while answering a query but it has to be restored to its original state afterwards.

Throughout the paper, it is assumed that $n$ is the number of items of the input each stored in a constant number of words, 
and that the available workspace is $\Theta(s)$ bits, where $\lg n \leq s \leq n \cdot \lg n$.
Since a single cursor, which is necessary to iterate over the input,
already needs $\lg n$ bits, we can not hope to solve any of the problems with less workspace.
In addition, and as usual, it is assumed that operations on the input coordinates can be performed in constant time each.
We emphasize that this assumption is not essential for our algorithms to work, but only scales with their running times.

Next, we survey some of the major results known for the read-only random-access model.
Pagter and Rauhe \cite{PagR98} gave an asymptotically-optimal algorithm for sorting $n$ elements, 
which runs in $O(n^2/s + n \cdot \lg s)$ time.
A simplified variation of this sorting algorithm is given in \cite{AsaEK13}.
Beame \cite{Bea91} established a matching $\Omega(n^2)$ lower bound for the time-space product for sorting in the stronger branching-program model.
Several papers \cite{ElmJKS14,Fre87,MunR96,RamR99} considered the selection problem in the read-only random-access model. Elmasry et al.~\cite{ElmKH14} introduced space-efficient algorithms for basic graph problems. 
Concerning geometric problems, Chan~\cite{Cha02} presented an algorithm for the
closest-pair problem with integer coordinates in the word RAM model,
and his algorithm can be made to work in the read-only model.
Darwish and Elmasry \cite{DarE14} gave an optimal planar convex-hull construction algorithm that runs in $O(n^2/s + n \cdot \lg s)$ time. Konagaya and Asano \cite{KonA13} gave an algorithm for reporting line-segments intersections that runs in $O((n^2/\sqrt{s}) \cdot \sqrt{\lg n} + k)$ time, where $k$ is the number of intersections.
Chan and Chen~\cite{ChaC07} have noted that the standard
Clarkson-Shor approach leads to randomized multi-pass streaming
algorithms for 3-D convex hulls and 2-D Voronoi diagrams as long as $s \ge n^{1/c}$ bits 
of working space are available, where $c$ is a fixed constant.
Recently, Korman et al.~\cite{KorMRRSS15} gave space-efficient algorithms 
for triangulations and for constructing Voronoi diagrams whenever $s =
\Omega(\lg n \cdot \lg \lg n)$ bits of working space are available.  
Asano et al.~\cite{AsaMRW11} considered space-efficient plane-sweep algorithms for Delaunay
triangulation and Voronoi diagram. However, they only considered the case where $s=\Theta(\log n)$ bits, 
and both algorithms run in $O(n^2)$ time for this case.
Other papers that deal with space-efficient geometric algorithms include \cite{AsaBBKMRS14,BarKLSS13}.

As a building block for our algorithms we use the {\em adjustable navigation pile} \cite{AsaEK13}; an efficient priority-queue-like data structure that uses $O(s)$ bits, where $\lg n \leq s \leq n \cdot \lg n$, in the read-only random-access model of computation. Given a read-only input array of $n$ elements and a specified value, an adjustable navigation pile can be initialized in $O(n)$ time. Subsequently, the elements that are larger than the given value can be streamed in sorted order in $O(n/s + \lg s)$ time per element. 
Thus, it is possible to stream the next $k$ elements starting with a specified value in sorted order in $O((n/s + \lg s) \cdot k + n)$ time, and all the elements of the array can be streamed in sorted order in $O(n^2/s + n \cdot \lg s)$ time.

Another ingredient that we use in some of our algorithms is a rank-select data structure \cite{Cla96,Mun96}. 
A rank-select data structure can be built on a vector of $n$ bits using $O(n)$ time and $o(n)$ extra bits,
and supports in $O(1)$ time the queries $\Rank(i)$, which returns the number of 1-bits in the first $i$ positions 
of the bit vector, and $\Select(j)$, which returns the index of the $j$-th 1-bit in the bit vector.
In accordance, one can sequentially scan the entries of the bit vector that have 1-bits in $O(1)$~time
per entry.

In this paper we give space-efficient plane-sweep algorithms for solving planar geometric problems.
In contrast to Asano et al.~\cite{AsaMRW11}, all our algorithms allow a
trade-off between time and space and work for all values of $s$ where $\lg n \leq s \leq n \cdot \lg n$.
In Section \ref{stretching} we introduce a general technique 
that we call the {\em stretching technique} relying on splitting the input array, 
and later employ it in our algorithms. 
In Section \ref{intersect} we give a simple algorithm for enumerating intersections among $n$ line segments that runs in $O((n^2/s^{2/3}) \cdot \lg s + k)$ time, where $k$ is the number of intersections. Our algorithm is asymptotically faster than that of Konagaya and Asano for all values of $s$. We point out that the same approach can be used to count the number of intersections in $O((n^2/s^{2/3}) \cdot \lg s)$ time. In Section \ref{closest} we give an algorithm for finding the closest pair among $n$ points whose running time is $O(n^2/s + n \cdot \lg s)$. 
To obtain this result, we combine new ideas with the classical plane-sweep
and divide-and-conquer approaches for solving the closest-pair problem.
A lower bound of $\Omega(n^{2-\epsilon})$ was shown by Yao \cite{Yao94} for the time-space product of the element-distinctness problem, where $\epsilon$ is an arbitrarily small positive constant. This lower bound applies for the closest-pair problem, indicating that our algorithm is close to optimal. In Section \ref{counting-axis-parallel} we give an algorithm for counting the intersections among $n$ axis-parallel 
line segments that runs in $O((n^2/s) \cdot \lg^{4/3} s + n^{4/3} \cdot \lg^{1/3} n)$ time.
The idea is to partition the plane as a grid and to run local plane sweeps
on parts of the plane with truncated segments. 
In Section \ref{batching} we sketch a so-called {\em batching technique} to represent the sweep line for special plane-sweep algorithms using fewer bits than usual, and then utilize this technique in Section \ref{enumerating-axis-parallel} for enumerating the intersections among $n$ axis-parallel line segments in $O((n^2/s) \cdot \lg s \cdot \lg \lg s + n \cdot \lg s + k)$ time, where $k$ is the number of intersections. In Section~\ref{measure} we show how to calculate Klee's measure (the area of the union) for $n$ axis-parallel rectangles in $O((n^2/s) \cdot \lg n + n \cdot \lg s)$ time if the corners of the rectangles are stored in sorted order. In Section~\ref{measure2} we introduce another general technique that we call the {\em multi-scanning technique} where we partition the plane and run several plane sweeps interleaved in a tricky way. 
We use this technique to calculate Klee's measure in $O((n^2/s + n \cdot \lg s) \cdot \sqrt{(n/s) \cdot \lg n})$ time if the corners of the rectangles are unsorted. 
We conclude the paper in Section \ref{comments} with some comments.

\section{A Stretching Technique: Splitting the Input Array}
\label{stretching}

We call a problem {\em decomposable} if its solution can be efficiently calculated by 
partitioning the input into smaller overlapping subsets, computing the partial solutions for 
these subsets, and combining these partial results to form the final outcome.
We also assume that the time needed to combine the results is shorter 
than that for computing the partial solutions.
Examples of such problems that we deal with in this paper are the closest-pair problem
and the axis-parallel line-segments intersections problem. 
For the closest-pair distance, the overall solution is the minimum among the partial solutions for the subproblems.
For the enumeration of the axis-parallel line-segments intersections, 
the overall solution is the union of the non-overlapping partial solutions.
The general line-segments intersections problem is also decomposable, 
and can be handled using the same approach with slight modifications.

The following technique allows us to stretch down the range of the workspace, for which we can 
efficiently solve a decomposable problem, to include smaller values.

Assume that the available workspace is enough to only handle a subset 
of the input that comprises $O(r)$ elements at a time, for some parameter $r \leq n$.  
Let $n$ be the number of elements in the input array.
Split the array into $\lceil n/r \rceil$ \emph{batches} $B_1,\ldots,B_{\lceil n/r \rceil}$
of at most $r$ consecutive elements each (the last batch may have less) and proceed as follows:
For $i=1,\ldots,{\lceil n/r \rceil}$ and $j=i+1,\ldots,{\lceil n/r \rceil}$,
apply the underlying algorithm within $B_i\cup B_j$.
Compute the overall answer by combining the partial results.
As we try all pairs of subproblems, the algorithm 
correctly explores all the possible subproblems $B_i\cup B_j$ for some $i$ and $j$, 
and accordingly produces the output correctly for decomposable problems.

The number of the subproblems handled in sequence is $\Theta(n^2/r^2)$.
Let the time needed to solve a subproblem of size $O(r)$ be $t(r)+k'$, it follows that the
overall time spent by the algorithm is $O((n^2/r^2) \cdot t(r) + k)$, where $k = \sum k'$.

\begin{lemma}
Suppose we know how to solve a decomposable problem $\mathcal{P}$ of size $n$ using
$s' = \Theta(f(n))$ bits in $O(n^2/g(s') + n \cdot \lg s')$ time,
where $f,g:\Nat \rightarrow \Rat$ are functions with $\lg n \leq f(n) \leq n \cdot \lg n$.
For any $s$ where $\lg n \leq s \leq s'$, we can solve any instance $I$ of $\mathcal{P}$ of size $n$ in
$O(n^2/g(s) + (n^2/f^{-1}(s)) \cdot \lg s)$ time with $O(s)$ bits.
In particular, when $f(n)=O(n/\lg n)$ and $g(s) = O(s)$, we can solve $I$ in $O(n^2/g(s))$ time and $O(s)$ bits.

\end{lemma}

\begin{proof}
By definition of $\mathcal{P}$, 
we can solve instances of $\mathcal{P}$ of 
size $r=\lceil f^{-1}(s) \rceil$ using $s$ bits in $t(r) = O(r^2/g(s) + r \cdot \lg s)$ time.
By applying the above construction, we can solve $I$ in $O((n^2/r^2) \cdot t(r)) = O(n^2/g(s) + (n^2/r) \cdot \lg s) = O(n^2/g(s) + (n^2/f^{-1}(s)) \cdot \lg s)$ time. 
If $f(n)=O(n/\lg n)$, then $f^{-1}(s) = \Omega(s \cdot \lg s)$, and we can solve $I$ in $O(n^2/g(s) + n^2/s)$.
If in addition $g(s) = O(s)$, the claimed time and space bounds follow.
\end{proof}

\section{Line-Segments Intersections}
\label{intersect}

Given a set of $n$ line segments in the plane, the line-segments-intersections 
problem is to enumerate all the intersection points among these line segments. 
The counting version of the problem is to produce the number of these intersections.
Given a set of red segments and another of blue segments, the bichromatic-intersections 
problem is the problem of reporting the intersections between red segments and blue segments.
An optimal algorithm to enumerate all the intersections that runs in $O(n \cdot \lg n + k)$ 
time was given by Balaban~\cite{Bal95}, where $n$ is the number of segments and $k$ is the number of intersections returned.
Chazelle \cite{Cha93}, improving over Agarwal \cite{Agr90}, showed how to count the intersections among 
$n$ line segments in $O(n^{4/3} \lg^{1/3} n)$ time, and how to report $k$ 
bichromatic intersections in $O(n^{4/3} \lg^{1/3}{n} + k)$ time.
All these algorithms require a linear number of words, i.e., $f(n)=O(n \cdot \log n)$ bits.

If the available workspace is $\Theta(s)$ bits with $\lg n \leq s \leq n \cdot \lg n$, 
we give next a straightforward application of the stretching technique. 
We can apply the reporting algorithms on batches of size $O(r)$ line segments, where 
we choose $r=\Theta(f^{-1}(s))$, i.e, $r = \Theta(s/ \lg s)$.
First we apply Balaban's algorithm for each batch separately, then we apply a bichromatic-intersections
algorithm on every pair of batches (assuming one of them is red and the other is blue).
Note that we cannot apply Balaban's algorithm on pairs of batches, 
for otherwise the partial solutions will be overlapping 
(intersections among the segments of a batch will show up in several partial solutions), 
and hence combining the partial solutions would be problematic.
It follows that $t(r) = O(r^{4/3} \cdot \lg^{1/3} r )$.
The reported intersections are the union of the non-overlapping intersections found by solving the subproblems.
Hence, the overall time for this algorithm is $O((n^2/r^2) \cdot t(r) + k) = O((n^2/r^{2/3}) \cdot \lg^{1/3} r +
k) = O((n^2/s^{2/3}) \cdot \lg s + k)$ time, where $k$ is the number of reported intersections.

In the same vein, one can adapt the counting algorithm for a space-efficient variant. The running time for the counting version will be $O((n^2/s^{2/3}) \cdot \lg s$).

\begin{lemma}
\label{stretch}
Given a read-only array of $n$ elements and $\Theta(s)$ bits of workspace, where $\lg n \leq s \leq n \cdot \lg n$, the planar line-segments-intersections enumeration problem can be solved 
in $O((n^2/s^{2/3}) \cdot \lg s + k)$ time, where $k$ is the number of intersections returned.
The counting version can be solved in $O((n^2/s^{2/3}) \cdot \lg s)$~time.
\end{lemma}
\section{Closest Pair}
\label{closest}

Given a set of $n$ points in the plane, the planar closest-pair problem is to identify a pair of
points that are closest to each other.

Assume for the moment that the available workspace is $\Theta(s)$ bits, where $\sqrt{n} \cdot \lg n \leq s \leq n \cdot \lg n$. 
In a first step, we produce the points in sorted order according to their $x$-coordinate values
using an adjustable navigation pile, and consider them in order in groups having
$\lceil s/\lg n \rceil$ points each (except possibly the last group). 
Call the vertical regions containing these groups the {\it vertical strips},
and call the boundary vertical lines separating the vertical strips the {\it vertical
separators}---if necessary, rotate the plane slightly.
We deal with the points in our workspace by storing, for each of these points, 
$O(\lg n)$ bits of the index of the point in the input array.
As for the vertical separators, we store in $O(\lg n)$ bits the index of the horizontally closest point to each separator.
Since there are at most $m = \lceil (n/s) \cdot \lg n \rceil$ vertical strips,
references to the $x$-coordinate values of all the vertical separators can be stored in $O((n/s) \cdot \lg^2 n)$ bits, which is $O(s)$ as long as $s = \Omega(\sqrt{n} \cdot \lg n)$. 
The entities of all the separators can then be simultaneously stored within the available workspace.  
Additionally, all the points of a vertical strip can fit in the available workspace.
Thus, a standard closest-pair algorithm \cite{CorLRS09} can be applied to identify the closest pair  
among the points of each vertical strip one after the other. 
We then find the pair with the minimum closest distance among all the vertical subproblems, and call this minimum distance $\delta$.

In a second step, 
we produce the points in sorted $y$-coordinate order using another adjustable navigation pile. 
We now use a standard idea from the divide-and-conquer algorithm for the closest-pair problem.
We only retain the points that lie within a horizontal distance $\delta$ from
any of the vertical separators and ignore the other points.
Call the selected points the {\it candidate points}.
We consider the candidate points in the $y$-coordinate order in groups having
$8 \cdot m$ points each (except the last group that may have less points). 
Call the horizontal regions containing these groups the {\it horizontal strips}.
Note that references to the points of a horizontal strip can be stored in $O((n/s) \cdot \lg^2 n) = O(s)$ bits, which can all fit in the available workspace.
We can then apply a standard closest-pair algorithm within the working storage to identify 
the closest pair among the candidate points of every two consecutive horizontal strips in order. 
Let $\delta'$ be the minimum closest distance among all the horizontal subproblems.
Finally, we return $\min(\delta,\delta')$ as the closest-pair distance.  

We prove next the correctness of the algorithm. 
We only have to show that the restriction to the points close to the
vertical separators in the second step is correct. We slightly generalize
the proof for the standard divide-and-conquer algorithm for the closest-pair problem.
Since the distance between any pair of points within a vertical strip is at least $\delta$, 
any point that is at horizontal distance more than $\delta$ from all the vertical separators 
can not be closer than $\delta$ to any other point.
We then need to only proceed with the candidate points that lie within a horizontal distance $\delta$ from
any of the vertical separators. Fix a candidate point $p$.
Given a specific vertical separator, for the candidate points above $p$ to be closer than $\delta$ to $p$ they must lie together with $p$ within a $2\delta \times \delta$ rectangle centered at the vertical separator.
Note that there could be at most $8$ points above $p$ within this rectangle whose distances to $p$ are less than $\delta$,
since 6 circles of diameter $\delta$ cover the whole rectangle and 
there can be at most one point in the left and the right circles as well as
at most two points in the middle circles.
For an illustration of this fact, see Fig.~\ref{fig:0}. 
\begin{figure}[ht]
    \begin{center}
       \scalebox{1}{\includegraphics{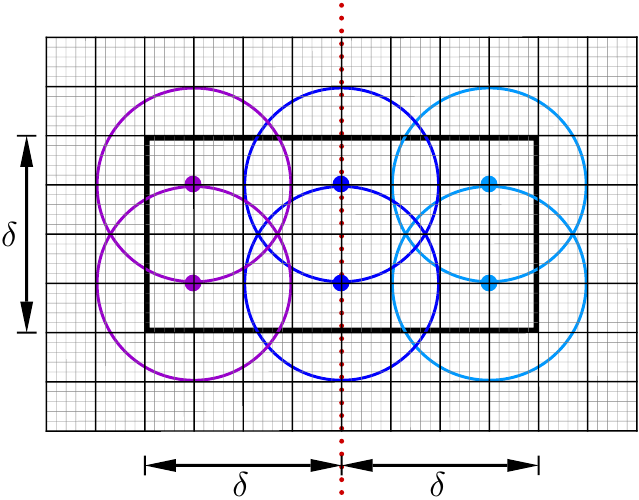}}
    \end{center}
    \caption{A vertical separator (dotted line) and a $2\delta \times
    \delta$ rectangle enclosing
    possible candidate points. The rectangle can be covered with 8 circles
    of diameter $\delta$.}
    \label{fig:0}
\end{figure}
Since there are $m$ separators, the number of candidate points above $p$ to be checked 
for possibly having a distance less than $\delta$ from $p$ is at most $8 \cdot m$; 
no other point above $p$ can be at distance less than $\delta$ from $p$.
(Actually, it suffices to check only $5$ candidate points above $p$ for each separator \cite[Exercise 33.4-2]{CorLRS09}.)
Obviously, the $p$-related candidate points must be consecutive in the $y$-coordinate values.
Since we store $8\cdot m$ candidate points per strip, the $p$-related candidate points above $p$ lie in only two horizontal strips, the horizontal strip that spans $p$ and the horizontal strip above it.
We conclude that we need to only consider the mutual distances 
among the points of each two consecutive horizontal strips. 

The time needed to produce the points in sorted order in both coordinates 
is bounded by the time for sorting using the adjustable navigation pile, which is $O(n^2/s + n \cdot \lg s)$ \cite{AsaEK13}.
The time needed to execute the standard closest-pair algorithm for all the strips
is $O(n \cdot \lg s)$ \cite{CorLRS09}. The time needed to check whether each point is close 
to one of the separators or not is $O(n \cdot \lg n) = O(n \cdot \lg s)$
using a binary search among the $x$-coordinates of the separators for each point. 
Hence, the overall running time of the algorithm is $O(n^2/s + n \cdot \lg s)$.   

Assume now that we have
$\Theta(s)$ bits available, where $\lg n \le s < \sqrt{n} \cdot \lg n$.
Let $r = s^2/ \lg^2 s$.
As $s= \Theta(\sqrt{r} \cdot \lg r)$, we can apply the
above algorithm on instances of size $\Theta(r)$.
In such a case, the running time for each of theses instance would be $t(r) = O(r^2/s + r \cdot \lg s) = O(r^2/s)$.
We then divide the input into $\lceil n/r \rceil$ batches of points and apply the stretching technique.
We compute the closest pair within every pair of batches, and return the overall closest pair.
The space needed is indeed $O(s)$, and the time consumed is $O(({n/r})^2 \cdot t(r)) = O(n^2/s)$.

\begin{lemma}
Given a read-only array of $n$ elements and $\Theta(s)$ bits of workspace, where $\lg n \leq s \leq n \cdot \lg n$, 
the planar closest-pair problem can be solved in $O(n^2/s + n \cdot \lg{s})$ time.
\end{lemma}

It is well-known that the closest-pair algorithm
can be generalized from two to higher
dimensions~\cite{CorLRS09} to run in $O(n \lg^{d-1} n)$ time in $d$ dimensions.
Applying the stretching technique in a similar way as described above, we can solve the closest-pair problem in $d$
dimensions with $\Theta(s)$ bits of workspace in $O(n^2/s + n \cdot \lg^{d-1}{s})$ time, 
where $\lg n \leq s \leq n \cdot \lg n$.
\section{Counting Axis-Parallel Line-Segments Intersections}
\label{counting-axis-parallel}

Given a set of $n$ axis-parallel (horizontal or vertical) line segments in the plane,
we want to count the intersection points among these line segments. 

Assume for the moment that the available workspace is $\Theta(s)$ bits, where $n^{2/3} \cdot \lg n \leq s \leq n \cdot \lg n$. 
First, we produce the endpoints of the line segments in sorted order according to their $x$-coordinate values
using an adjustable navigation pile, and consider them in order in groups having
$\lceil s/\lg n \rceil$ points each (except possibly the last group that may have fewer points). 
Again, call these groups the {\em vertical strips},
and call the boundary lines separating the strips the {\em separators}.
Since there are $\lceil (n/s) \cdot \lg n \rceil = O(n^{1/3})$ vertical strips,
references to the $x$-coordinate values of all the separators can be stored within the workspace.  
We associate a line segment to a strip if at least one of its two endpoints lie inside the strip.
The references to the line segments of a vertical strip can all fit in the workspace.
We can then apply a standard line-segments-intersections counting algorithm to each vertical strip one after the other, 
and add these counts together. See the left side of Fig.~\ref{fig:1}.

\begin{figure}[b!]
    \begin{center}
       \scalebox{0.9}{\includegraphics{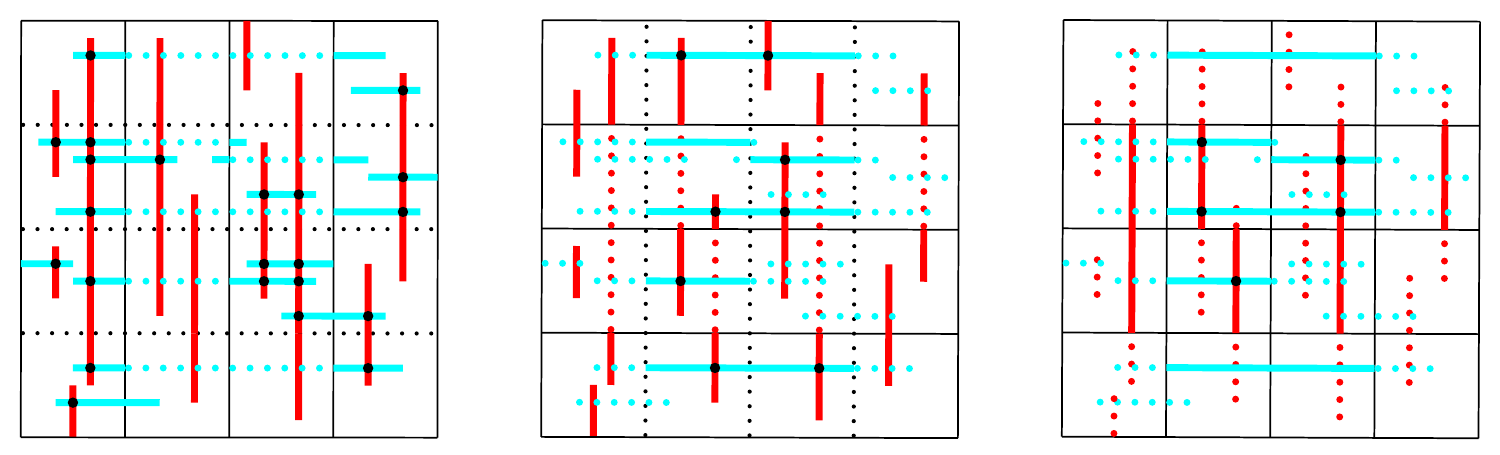}}
    \end{center}
    \caption{Counting axis-parallel line segments in three {\em phases}. The black dots
    on the crossings of two segments show the intersection points that are counted in
    each phase.}
    \label{fig:1}
  \end{figure}

Subsequently, we produce the points in sorted order according to their $y$-coordinate values
using another adjustable navigation pile, and consider them in {\em
horizontal strips} having $\lceil s/\lg n \rceil$ points each (except possibly the last group that may have less points). 
It follows that the $O(n^{1/3})$ references to the so-called {\em horizontal separators} can be simultaneously 
stored in the workspace. In a similar fashion as above, 
we apply a line-segments-intersections counting algorithm to each horizontal strip one after the other, 
and add these counts to the accumulated count.  
To avoid counting intersections twice, we truncate the horizontal segments while dealing with them 
such that each new endpoint lies on the closest vertical separator to the old endpoint intersecting the segment.
Note that the intersections of the truncated parts of the horizontal segments with vertical segments
have already been accounted for while dealing with the vertical strips. 
See the middle of Fig.~\ref{fig:1}.

Let $\mathcal{R}_{i,j}$ be the {\em cell} formed by the intersection of 
the $i$th horizontal strip with the $j$th vertical strip. 
A line segment {\em spans} a cell if it crosses two of the cell's boundaries. 
What is left is to account for the intersections among these spanning segments and add it to the accumulated counts.
We show next how to count the spanning horizontal segments for each cell. The treatment for the vertical segments is similar.
A line segment is {\em interior} to a cell if both its endpoints lie inside the cell.
For each cell $\mathcal{R}_{i,j}$, we store the count $b_{i,j}$ of horizontal segments beginning in
the cell, the count $f_{i,j}$ of horizontal segments finishing in the cell, and the count $t_{i,j}$ of the horizontal segments interior to the cell. Since there are $O(n^{2/3})$ cells, all these values can be stored in $O(n^{2/3} \cdot \lg n)$ bits, which is $O(s)$ when $s \geq n^{2/3} \cdot \lg n$.
For every horizontal segment, we locate the starting and ending cells using binary search among the separators,
and increment the corresponding counters in accordance. 
We then scan the cells of every horizontal strip sequentially while calculating $e_{i,j}$ 
the number of horizontal segments {\em entering} $\mathcal{R}_{i,j}$, 
i.e., the number of segments that have a non-empty intersection with
$\mathcal{R}_{i,j}$ and $\mathcal{R}_{i,j-1}$;
this is done using $e_{i,0} = 0$ and $e_{i,j} = e_{i,j-1} + b_{i,j-1} - f_{i,j-1}$.
We then compute the number of horizontal segments spanning $\mathcal{R}_{i,j}$ as $e_{i,j} - f_{i,j} + t_{i,j}$.
The number of intersections of the spanning segments of $\mathcal{R}_{i,j}$ is the product of
its spanning horizontal and vertical segments. 
We finally sum the counts of these intersections for all cells. 
See the right side of Fig.~\ref{fig:1}.  
 
The time needed to produce the endpoints in sorted order in both coordinates 
using the adjustable navigation pile is $O(n^2/s + n \cdot \lg s)$ \cite{AsaEK13}.
The time needed to execute the standard segments-intersection counting algorithm for all the strips
is $O(n^{4/3} \cdot \lg^{1/3} n)$. The time needed to perform binary search among the separators is $O(n \cdot \lg s)$.
The time needed to count the intersections of the spanning segments of all the cells is constant per cell 
and sums up to $O(n^{2/3})$.
It follows that the overall running time of the algorithm is $O(n^{4/3} \cdot \lg^{1/3} n)$.   

Assume now that we have $\Theta(s)$ bits available, where $\lg n \le s < n^{2/3} \cdot \lg n$.
Let $r = s^{3/2}/ \lg^{3/2} s$. Since $s=\Theta(r^{2/3}\lg r)$,
we can apply the above algorithm on instances of $\Theta(r)$ elements.
We divide the input array into $\lceil n/r \rceil$ batches of consecutive segments and apply the stretching technique.
First apply the algorithm on instances for every batch individually, then on instances for every pair of batches.
Using these computed counts, the overall count can be easily calculated.
The running time for each instance would be $t(r) = O(r^{4/3} \cdot \lg^{1/3} s)$.
The overall time consumed in this case is $O(({n/r})^2 \cdot t(r)) = O((n^2/s) \cdot \lg^{4/3} s)$.

\begin{lemma}
Given a read-only array containing the endpoints of $n$ line segments and $\Theta(s)$ bits of workspace, 
where $\lg n \leq s \leq n \cdot \lg n$, counting the planar axis-parallel line-segments intersections can 
be done in $O((n^2/s) \cdot \lg^{4/3} s + n^{4/3} \cdot \lg^{1/3} n)$ time.
\end{lemma}

\section{A Batching Technique: Processing Sweep-Line Events in Batches}
\label{batching}

Plane sweep is one of the most common algorithmic techniques in
computational geometry.  The idea is to move a line across
the plane and to maintain the intersection of that line with the
objects of interest. Many geometric problems have been solved
using this paradigm~\cite{BenO79,BerOKO08,CorLRS09}.
We assume that the sweep line moves over the plane from left to right.
Only at particular {\em event points} is an update of the status required.
Typically, a plane-sweep algorithm uses a priority queue ({\em event queue}) to produce the upcoming events in order 
and a balanced binary search tree ({\em status structure}) to store and query the objects 
that cross the sweep line in order.
Since $\Theta(n)$ objects might be part of the search tree, 
a standard plane-sweep algorithm needs $\Theta(n \cdot \lg n)$ bits.

In the following we show that, if the given objects are axis-parallel, one
may reduce the working storage of the status structure to $\Theta(n)$ bits
by processing the events in batches. 
The ideas of this technique are sketched next.

Suppose our plane is divided into $m$ vertical and horizontal strips 
such that each strip contains $O(n/m)$ local objects, where an object is {\em local} 
for a strip if it starts or ends within the strip. As before, the boundaries
of the strips are called {\em separators}.
The intersection of a horizontal strip with a vertical strip is called a {\em cell}.
To apply the batching technique, we need the following two properties: 
(1) All the events of the event queue are on vertical separators,
i.e., they result from so-called {\em horizontally spanning objects}.
(2) All the objects of the status structure start and end on horizontal separators, 
i.e., they are so-called {\em vertically spanning objects}.
The situation is illustrated in Fig.~\ref{fig:2}.
Assume the available workspace is $\Theta(s)$ bits, where $n \leq s \leq n \cdot \lg n$. 
By setting $m = \lceil (n/s) \cdot \lg n \rceil$, we 
can store references to all local objects of a strip and references to the coordinates of the separators in the working storage.

\begin{figure}[h!]
    \begin{center}
       \scalebox{0.9}{\includegraphics{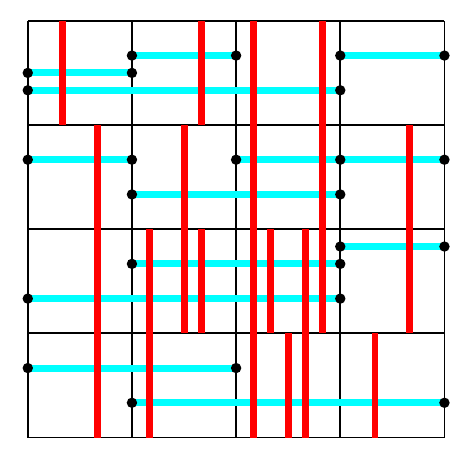}}
    \end{center}
    \caption{The plane is partitioned into cells with vertical and
    horizontal separators. The events are all on vertical separators and are
    shown through black dots.}
    \label{fig:2}
  \end{figure}

Because of (1) and (2), it is enough to update
the status structure only once per vertical strip with a batch of objects.
To 'represent' the status structure, we split the vertical strip to $m$ cells 
formed by the intersections with the horizontal strips.
We store the indices of the $O(n/m)$ 
vertically spanning objects of the cells of the vertical strip 
in an array using a total of $O((n/m) \cdot \lg n) = O(s)$ bits. In addition, we store for each 
of the $m$ cells a bit vector of $O(n/m)$ bits indicating whether each of these
objects spans the cell or not.
Over and above, for each bit vector of a cell, we build a rank-select data structure
that allows us to scan the vertical spanning objects of the cell in constant time per object. 
The bit vectors and the rank-select structures are enough to represent the status structure.
This way, the sweep line can be stored in a total of $O(s)$ bits.
 
We use an adjustable navigation pile as our event queue to produce the events and the spanning objects in order.
Since $s \geq n$, the time to produce all the events in order throughout the procedure is $O(n \cdot \lg s)$.
When the sweep line moves to a new vertical strip, 
we update the representation of the status structure as follows:
The vertical spanning objects in the new strip are produced by the navigation pile.
For each such object, the cells it spans are allocated in $O(m)$ time per object
by simply comparing the object coordinates with the horizontal separators.
The bit-vectors entries and the rank-select structures are updated accordingly.
The time to update the status structure (build a new one) is $O(n)$. 
Throughout the algorithm, the total time to update the status structure
is $O(n\cdot m) = O((n^2/s) \cdot \lg n) = O((n^2/s) \cdot \lg s)$.

What is left is to show how to allocate an event point within the status structure representing the sweep line.
We would be satisfied with only identifying the cell that contains this event point within the vertical strip.
We do that using binary search against the $m$ horizontal separators,
consuming $O(\lg m) = O(\lg \lg s)$ time per event point.

\begin{lemma}
Using the batching technique, a sweep can be performed on a plane with $n$ objects,
using a data structure that can be stored in $\Theta(s)$ bits, where $n \leq s \leq n \cdot \lg n$.
The sweep makes a total of $O((n/s) \cdot \lg s)$ stopovers, and the data structure can be rebuilt in $O(n)$ time per stopover
plus a total of $O(n \cdot \lg s)$ time,
and queried in $O(\lg \lg s)$ time per event. Handling all events at each stopover,
the total time consumed is $O((n^2/s) \cdot \lg s \cdot \lg \lg s + n \cdot \lg s)$.
\end{lemma}

\section{Enumerating Axis-Parallel Line-Segments Intersections}
\label{enumerating-axis-parallel}

Assume for the moment that the available workspace is $\Theta(s)$ bits, where $n \leq s \leq n \cdot \lg n$.

We use the same ideas as our counting algorithm. 
As before, we split the plane into $m$ horizontal and vertical strips, where $m=\lceil (n/s) \cdot \lg n\rceil$.
Accordingly, we enumerate the intersections among the local parts of the segments within the strips by applying a standard line-segments-intersection enumeration algorithm.   

By truncating the segments, we assume from now on that all the endpoints 
lie on the boundaries of the cells and the segments span the cells they cross.
Note that each horizontal line segment that spans a cell must intersect 
all the vertical segments spanning the same cell. 
By applying the ideas of the batching technique, we store the vertical spanning line segments that 
lie in the current vertical strip and build a status structure that consumes $\Theta(s)$ bits in $O(n)$ time. 
Using this data structure it is possible to enumerate the vertical segments that span a given cell 
in time proportional to the number of the reported segments. 
For each horizontal segment, we check if it spans any of the cells of the sweep line.
We do that using binary search for each horizontal segment against the $m$ horizontal separators.
After every binary search for a horizontal segment, we query the status structure  
to find the vertical segments spanning the same cell, and
so their intersections with the horizontal segment are computed and reported. 
After locating the crossing cells of all the horizontal segments with the sweep line, 
the sweep line is advanced to the next vertical strip.

The total time needed to execute the standard algorithm locally within all the strips
is $O(n \cdot \lg s)$, which matches the time bound to build the status structure 
of all vertical strips using the batching technique. 
The time to perform binary search for each of the horizontal segments against the $m$ 
cells of the status structure is $O(n \cdot \lg m)$.
Hence, we can compute all intersection points of a vertical strip 
in $O(n \cdot \lg m + k')$ time, where $k'$ is the number of these intersections.
Since we repeat these actions for every vertical strip as the sweep line advances, the 
total time  is $O(n \cdot m \cdot \lg m + k) = O((n^2/s) \cdot \lg s \cdot \lg \lg s
 + k)$, where $k$ is the number of intersections returned.   
Since we can partition the  plane into strips using a navigation pile in
$O(n^2/s + n \cdot \lg s)$ time, the total time consumed by the whole algorithm is 
$O((n^2/s) \cdot \lg s \cdot \lg \lg s + n \cdot \lg s + k)$. 

Assume next that we have $\Theta(s)$ bits, where $\lg n \le s < n$. Let $r = s$.
We can then apply the above algorithm on instances of size $\Theta(r)$.
In such a case, the running time for each instance would be $t(r)+k' = O((r^2/s) \cdot \lg s \cdot \lg \lg s + r \cdot \lg s +
k') = O((r^2/s) \cdot \lg s \cdot \lg \lg s + k')$, where $k'$ is the number of intersections.
We then divide the input into $\lceil n/r \rceil$ batches of segments and apply the stretching technique
on pairs of batches, a batch of vertical segments with a batch of horizontal segments.
The total time consumed is $O(({n/r})^2 \cdot t(r) + k) = O((n^2/s) \cdot \lg s \cdot \lg \lg s +
k)$, where $k$ is the number of intersections returned.   
 
\begin{lemma}
Given a read-only array containing the endpoints of $n$ line segments and $\Theta(s)$ bits of workspace, where $\lg n \leq s \leq n \cdot \lg n$, enumerating the planar axis-parallel line-segments intersections is done in $O((n^2/s) \cdot \lg s \cdot \lg \lg s + n \cdot \lg{s} + k)$ time, where $k$ is the number of intersections returned.
\end{lemma}

\section{Measure of Axis-Parallel Rectangles}
\label{measure}

We consider the problem of computing the {\em measure} of a set of
$n$ axis-parallel rectangles, i.e., the size of the area of the union.
The problem was posed by V. Klee, and thus called {\em Klee's measure problem}.
Bentley~\cite{Ben77} described an $O(n \cdot \lg n)$-time
algorithm that can be implemented with $\Theta(n \cdot \lg n)$ bits of working space.
Bentley's algorithm sweeps a vertical line from left to right across the
rectangles and maintains the intersection of the rectangles and the sweep line.
Another algorithm to compute the measure was presented by Overmars and Yap~\cite{OveY88}. A generalization of the algorithm to $d$ dimensions was given by Chan~\cite{Cha13}.

Assume that the available workspace is $\Theta(s)$ bits, where $\lg n \leq s \leq n \cdot \lg n$. 
To compute the measure of a set of $n$ axis-parallel rectangles, we use Bentley's algorithm as a subroutine. 
In this section, we restrict ourselves to the case where the corners of the rectangles are stored sorted by their $x$-coordinates. This restriction is dropped in the next section.

We split the plane into $m=\Theta((n/s) \cdot \lg n)$ horizontal strips,
where each strip consists of $\Theta(s/ \lg n)$ rectangle corners 
and accordingly fit in the available workspace.
A rectangle is {\em spanning} a strip if its vertical segments cross the two separators of the strip.
We process the strips in sorted $y$-coordinate order, one after the other.
By using an adjustable navigation pile, we produce and store the
rectangles cornered within each strip in sequence. 
Before processing a strip and storing the rectangles, 
we truncate those rectangles such that they are shrunk to their intersection with the strip.
We would then run Bentley's algorithm on these rectangles. 
However, we need to also take into consideration the rectangles spanning the strip.
We show next how to do that efficiently.

We horizontally scan the strip and keep track of the spanning segments and the corners.
We accumulate as a global variable the width $\mathcal{W}$ of the union of the spanning rectangles so far.
To do that, we maintain $z$ as the difference between the 
number of scanned spanning segments that are left boundaries of a rectangle and the number of
scanned spanning segments that are right boundaries.
Whenever $z$ becomes positive, we record this coordinate as $x_1$. 
Whenever $z$ returns back to zero, we record this coordinate as $x_2$; 
we have just passed over a spanning area, and accordingly update $\mathcal{W}$ by adding to it the value $x_2 - x_1$.
Whenever we meet a corner, we update its $x$-coordinate value as follows.
If $z$ is positive (the corner is in a spanning area), first set the $x$-coordinate of this corner to $x_1$.
Either way, whether $z$ is positive or zero, we subtract the current value 
of $\mathcal{W}$ from the $x$-coordinate of the corner.  
This process of relocating the corners is called {\em simplifying} the rectangles in~\cite{Cha13}.
After finishing the scan, we apply Bentley's algorithm to 
the relocated corners and calculate the measure within the current strip.
We also multiply $\mathcal{W}$ by the width of the strip to get the area covered by the spanning rectangles,
and add this area to the calculated measure.
The total measure is the sum of the measures within all the strips.

The time to sequentially scan the segments and simplify the rectangles within each strip is $O(n)$ (as the segments are already sorted), and the time for applying Bentley's algorithm is $O((s/\lg n) \cdot \lg s)$.
In accordance, the total time to process all the $m$ strips is $O((n^2/s) \cdot \lg n + n \cdot \lg s)$.

\begin{lemma}
Given a read-only array storing the corners of $n$ axis-parallel rectangles in sorted
$x$-coordinate order, and the available workspace is $\Theta(s)$ bits, where $\lg n \leq s \leq n \cdot \lg n$, 
the measure (area of the union) can be computed in $O((n^2/s) \cdot \lg n + n \cdot \lg s)$ time.
\end{lemma}

\section{A Multi-Scanning Technique: Partitioning the Plane}
\label{measure2} 

In this section we introduce a general technique that can be used in different space-efficient algorithms,
and apply it to the measure problem if the input is not sorted.
The idea is to perform alternating vertical and horizontal sweeps on parts of the plane to identify cells, each containing a set of objects that fit in the working storage. Once identified, we apply a local algorithm within each cell.
By partitioning the plane into a grid of cells, we combine the local solutions for the cells together 
to obtain the final outcome. The details come next. 

We partition the plane into $m = \lceil \sqrt{(n/s) \cdot \lg n} \rceil$ horizontal strips, 
where each strip consists of $O(n/m)$ corners. 
We process the horizontal strips one after the other in sorted $y$-coordinate order using an adjustable navigation pile.
Once the two separators of a horizontal strip $\mathcal{H}$ are determined, 
we initialize an adjustable navigation pile $Y_\mathcal{H}$ for the strip that allows us to
stream the corners within $\mathcal{H}$ ordered by their $y$-coordinates.
We start sweeping over the plane in sorted $x$-coordinate order using another 
adjustable navigation pile $X_\mathcal{H}$ that is initialized over the whole input.
For this horizontal sweep, we are interested only in the corners in $\mathcal{H}$ as well
as the vertical segments spanning $\mathcal{H}$---to find the spanning
segments, we have to take all corners of the plane into consideration.
Whenever the number of corners in $\mathcal{H}$ produced by $X_\mathcal{H}$ 
is $\ell = \lceil s/\lg n \rceil$ (except for the last cell that may have less corners), 
we have reached a vertical separator that identifies, as a right boundary, a cell $\mathcal{V}$ within $\mathcal{H}$. 
The corners of a cell can be stored in $O(s)$ bits and hence fit in the working storage.
During this horizontal sweep over $\mathcal{V}$, we calculate the horizontal width $\mathcal{W}_h$ of the area 
covered by the vertically spanning rectangles, and in the meantime simplify 
these corners of $\mathcal{V}$ (relocate the $x$-coordinates), as explained 
in the previous section, while storing them.
We temporarily pause the horizontal sweep, and start a vertical sweep within $\mathcal{H}$
after initializing $Y_\mathcal{H}$ using the value of the horizontal separator between $\mathcal{H}$ and the strip above it. During this vertical sweep, we calculate the vertical width $\mathcal{W}_v$ of the area covered by the horizontally spanning rectangles, and simplify the stored corners of $\mathcal{V}$ (this time, relocate the $y$-coordinates).
Since the corners within $\mathcal{V}$ fit in the working storage, we
compute Klee's measure of the parts of the simplified rectangles within
the cell $\mathcal{V}$ using Bentley's algorithm.  
We add the areas covered by the spanning vertical and the spanning horizontal rectangles to
adjust the measure, and subtract the intersection area $\mathcal{W}_h \times \mathcal{W}_v$ that has been added twice.
We repeatedly proceed with the horizontal 
sweep using $X_\mathcal{H}$ to identify and partially process a cell,
then alternately initialize $Y_\mathcal{H}$ and perform a vertical sweep 
within $\mathcal{H}$ to finish the processing of the cell. 
After all the cells of a horizontal strip are processed, 
we repeat the same actions for the next horizontal strips in sequence.
Since we correctly calculate the measure within every cell, the overall sum of all
the local measures is what we are looking for.

Concerning the running time, we consider the time to produce the segments by the navigation piles. 
Recall that $X_\mathcal{H}$ sweeps over all the $n$ corners, 
whereas $Y_\mathcal{H}$ sweeps only over the $O(n/m)$ corners of $\mathcal{H}$.
The navigation piles $X$ for the horizontal sweeps repeatedly process all the
input for every horizontal strip. Since we have a total of $m$ such sweeps, 
the total time consumed by the $X$ navigation piles is $O((n^2/s + n \cdot \lg s) \cdot m)$.
The navigation piles $Y$ for the vertical sweeps process the $O(n/m)$ 
corners of a horizontal strip in one sweep. 
Therefore, the total time for each of these vertical sweeps is $O((n/s + \lg s) \cdot n/m + n)$.
It is straightforward to verify that $n/s + \lg s = \Omega(m)$ for all considered values of $n$ and $s$
(it is either true that $n/s > m$ or otherwise $\lg s = \Omega(m))$. 
The total number of vertical sweeps done within each horizontal strip is $O((n/m)/\ell)$, which is
$O(m)$ since $m = \lceil \sqrt{(n/s) \cdot \lg n} \rceil$.
It follows that the total time of the vertical sweeps within one horizontal strip is 
$O(n^2/s + n \cdot \lg s)$. Multiplying by the number of horizontal strips $m$,
the total time consumed by the $Y$ navigation piles is $O((n^2/s + n \cdot \lg s) \cdot m)$,
matching the bound for the $X$ piles.
The time needed by the extended local version of Bentley's algorithm within
each cell is $O(\ell \cdot \lg \ell)$, resulting in a total of $O(n \cdot \lg s)$ time
for all the calls to Bentley's algorithm.  The time for the navigation piles is dominating.

\begin{lemma}
Given a read-only array containing the corners of $n$ axis-parallel rectangles,
and the available workspace is $\Theta(s)$ bits, where $\lg n \leq s \leq n \cdot \lg n$, 
the measure can be computed in $O((n^2/s + n \cdot \lg s) \cdot \sqrt{(n/s) \cdot \lg n})$ time.
\end{lemma}

\section{Concluding Comments}
\label{comments}

We have given space-efficient plane-sweep algorithms for some basic geometric problems.
We believe that the techniques we introduce cover a range of ideas to handle many other plane-sweep algorithms
in a space-efficient manner. We also believe that our techniques can be easily extended to higher dimensions.

Except for $\epsilon$ in the $\Omega(n^{2-\epsilon})$ Yao's lower bound for the element-distinctness problem, the $O(n^2/s + n \cdot \lg s)$ bound for the running time of the problems we solve would be optimal.   
It is an intriguing open problem to get rid of this $\epsilon$.

Another question is if it is possible to get around with the extra logarithmic factors in the running times
of the problem of enumerating the general and the axis-parallel line-segments intersections.
It also remains open if it is possible to solve the measure problem more efficiently when the input is not sorted.

\bibliographystyle{abbrv}
\bibliography{bibliography}

\begin{thebibliography}{10}

\bibitem{Agr90}
P.~K. Agarwal.
\newblock Partitioning arrangements of lines {II}: Applications.
\newblock {\em Discrete Comput. Geom.}, 5(6):533--573, 1990.

\bibitem{AsaBBKMRS14}
T.~Asano, K.~Buchin, M.~Buchin, M.~Korman, W.~Mulzer, G.~Rote, and A.~Schulz.
\newblock Reprint of: Memory-constrained algorithms for simple polygons.
\newblock {\em Comput. Geom. Theory Appl.}, 47(3, {P}art B):469--479, 2014.

\bibitem{AsaEK13}
T.~Asano, A.~Elmasry, and J.~Katajainen.
\newblock Priority queues and sorting for read-only data.
\newblock In {\em Proc. 10th International Conference on Theory and
  Applications of Models of Computation ({TAMC} 2013)}, volume 7876 of {\em
  LNCS}, pages 32--41. Springer, 2013.

\bibitem{AsaMRW11}
T.~Asano, W.~Mulzer, G.~Rote, and Y.~Wang.
\newblock Constant-work-space algorithms for geometric problems.
\newblock {\em J. Comput. Geom.}, 2(1):46--68, 2011.

\bibitem{Bal95}
I.~J. Balaban.
\newblock An optimal algorithm for finding segments intersections.
\newblock In {\em Proc. 11th Symposium on Computational Geometry}, pages
  211--219, 1995.

\bibitem{BarKLSS13}
L.~Barba, M.~Korman, S.~Langerman, R.~I. Silveira, and K.~Sadakane.
\newblock Space-time trade-offs for stack-based algorithms.
\newblock In {\em Proc. 30th International Symposium on Theoretical Aspects of
  Computer Science ({STACS} 2013)}, volume~20 of {\em LIPIcs}, pages 281--292.
  Schloss Dagstuhl -- Leibniz-Zentrum f\"ur Informatik, 2013.

\bibitem{Bea91}
P.~Beame.
\newblock A general sequential time-space tradeoff for finding unique elements.
\newblock {\em {SIAM} J. Comput.}, 20(2):270--277, 1991.

\bibitem{Ben77}
J.~L. Bentley.
\newblock Algorithms for {K}lee's rectangle problems, 1977.
\newblock Unpublished manuscript.

\bibitem{BenO79}
J.~L. Bentley and T.~Ottmann.
\newblock Algorithms for reporting and counting geometric intersections.
\newblock {\em IEEE Trans. Computers}, 28(9):643--647, 1979.

\bibitem{BerOKO08}
M.~d. Berg, O.~Cheong, M.~v. Kreveld, and M.~Overmars.
\newblock {\em Computational Geometry: Algorithms and Applications}.
\newblock Springer-Verlag TELOS, Santa Clara, CA, USA, 3rd ed. edition, 2008.

\bibitem{BroIKMMT04}
H.~Br{\"o}nnimann, J.~Iacono, J.~Katajainen, P.~Morin, J.~Morrison, and
  G.~Toussaint.
\newblock Space-efficient planar convex hull algorithms.
\newblock {\em Theor. Comput. Sci.}, 321(1):25--40, 2004.

\bibitem{Cha02}
T.~M. Chan.
\newblock Closest-point problems simplified on the {RAM}.
\newblock In {\em Proceedings of the 13th Annual {ACM-SIAM} Symposium on
  Discrete Algorithms ({SODA} 2002)}, pages 472--473. {ACM/SIAM}, 2002.

\bibitem{Cha13}
T.~M. Chan.
\newblock {K}lee's measure problem made easy.
\newblock In {\em Proc. 54th Annual {IEEE} Symposium on Foundations of Computer
  Science, {FOCS 2013}}, pages 410--419. {IEEE} Computer Society, 2013.

\bibitem{ChaC07}
T.~M. Chan and E.~Y. Chen.
\newblock Multi-pass geometric algorithms.
\newblock {\em Discrete Comput. Geom.}, 37(1):79--102, 2007.

\bibitem{ChaMR14}
T.~M. Chan, J.~I. Munro, and V.~Raman.
\newblock Selection and sorting in the ``restore'' model.
\newblock In {\em Proc. 25th Annual {ACM-SIAM} Symposium on Discrete Algorithms
  ({SODA} 2014)}, pages 995--1004. {SIAM}, 2014.

\bibitem{Cha93}
B.~Chazelle.
\newblock Cutting hyperplanes for divide-and-conquer.
\newblock {\em Discrete Comput. Geom.}, 9(2):145--158, 1993.

\bibitem{Cla96}
D.~Clark.
\newblock {\em Compact Pat Trees}.
\newblock PhD thesis, University of Waterloo, Waterloo, Ontario, Canada, 1996.

\bibitem{CorLRS09}
T.~H. Cormen, C.~E. Leiserson, R.~L. Rivest, and C.~Stein.
\newblock {\em Introduction to Algorithms}.
\newblock The MIT Press, 3rd edition, 2009.

\bibitem{DarE14}
O.~Darwish and A.~Elmasry.
\newblock Optimal time-space tradeoff for the {2D} convex-hull problem.
\newblock In {\em Proc. 22nd Annual European Symposium on Algorithms ({ESA}
  2014)}, volume 8737 of {\em LNCS}, pages 284--295. Springer, 2014.

\bibitem{ElmJKS14}
A.~Elmasry, D.~D. Juhl, J.~Katajainen, and S.~R. Satti.
\newblock Selection from read-only memory with limited workspace.
\newblock {\em Theor. Comput. Sci.}, 554:64--73, 2014.

\bibitem{ElmKH14}
A.~Elmasry, F.~Kammer, and T.~Hagerup.
\newblock Space-efficient basic graph algorithms.
\newblock In {\em Proc. 32nd Annual Symposium on Theoretical Aspects of
  Computer Science ({STACS} 2015)}, LIPIcs, pages 288--301, 2015.

\bibitem{Fre87}
G.~N. Frederickson.
\newblock Upper bounds for time-space trade-offs in sorting and selection.
\newblock {\em J. Comput. Syst. Sci.}, 34(1):19--26, 1987.

\bibitem{KonA13}
M.~Konagaya and T.~Asano.
\newblock Reporting all segment intersections using an arbitrary sized work
  space.
\newblock {\em IEICE Transactions on Fundamentals of Electronics,
  Communications and Computer Sciences}, 96-A(6):1066--1071, 2013.

\bibitem{KorMRRSS15}
M.~Korman, W.~Mulzer, A.~van Renssen, M.~Roeloffzen, P.~Seiferth, and Y.~Stein.
\newblock Time-space trade-offs for triangulations and {V}oronoi diagrams.
\newblock In {\em Proc. 14th Algorithms and Data Structures Symposium ({WADS}
  2015)}, 2015.

\bibitem{Mun96}
J.~I. Munro.
\newblock Tables.
\newblock In {\em PROC 16th FSTTCS}, volume 1180 of {\em LNCS}, pages 37--42.
  Springer, 1996.

\bibitem{MunP80}
J.~I. Munro and M.~S. Paterson.
\newblock Selection and sorting with limited storage.
\newblock {\em Theor. Comput. Sci.}, 12(3):315--323, 1980.

\bibitem{MunR96}
J.~I. Munro and V.~Raman.
\newblock Selection from read-only memory and sorting with minimum data
  movement.
\newblock {\em Theor. Comput. Sci.}, 165(2):311--323, 1996.

\bibitem{OveY88}
M.~H. Overmars and C.~Yap.
\newblock New upper bounds in {K}lee's measure problem (extended abstract).
\newblock In {\em Proc. 29th Annual Symposium on Foundations of Computer
  Science ({FOCS} 1988)}, pages 550--556. {IEEE} Computer Society, 1988.

\bibitem{PagR98}
J.~Pagter and T.~Rauhe.
\newblock Optimal time-space trade-offs for sorting.
\newblock In {\em Proc. 39th Annual {IEEE} Symposium on Foundations of Computer
  Science (FOCS 1998)}, pages 264--268. {IEEE} Computer Society, 1998.

\bibitem{RamR99}
V.~Raman and S.~Ramnath.
\newblock Improved upper bounds for time-space trade-offs for selection.
\newblock {\em Nord. J. Comput.}, 6(2):162--180, 1999.

\bibitem{Yao94}
A.~C.-C. Yao.
\newblock Near-optimal time-space tradeoff for element distinctness.
\newblock {\em {SIAM} J. Comput.}, 23(5):966--975, 1994.

\end{thebibliography}

\end{document}